\documentclass[12pt]{iopart}

\usepackage{iopams}

\usepackage{amssymb}
\usepackage{tensor}
\usepackage{amsthm}
\usepackage{setstack}
\usepackage{bm,upgreek}

\newcommand{\Hi}{\mathcal{H}}

\newcommand{\dell}{\partial}

\newcommand{\N}{\mathbb{N}}

\newcommand{\norm}[1]{\left| #1\right|}

\newtheorem{prop}{Proposition}

\begin{document}

\title[Information geometric approach to mixed state quantum estimation]{Information geometric approach to mixed state quantum estimation}

\author{Gabriel F. Magno$^1$, Carlos H. Grossi$^2$, Gerardo Adesso$^3$ and Diogo O. Soares-Pinto$^1$}
\address{$^1$ Instituto de F\'{i}sica de S\~{a}o Carlos, Universidade de S\~{a}o Paulo, CP 369, 13560-970, S\~{a}o Carlos, SP, Brazil}
\address{$^2$ Departamento de Matem\'{a}tica, ICMC, Universidade de S\~{a}o Paulo, Caixa Postal 668, 13560-970,
S\~{a}o Carlos, SP, Brazil}
\address{$^3$ School of Mathematical Sciences and Centre for the Mathematics and Theoretical Physics of Quantum Non-Equilibrium Systems, University of Nottingham, University Park, Nottingham NG7 2RD, United Kingdom} \ead{gabriel.magno@usp.br, grossi@icmc.usp.br, gerardo.adesso@nottingham.ac.uk, dosp@usp.br}

\vspace{10pt}
\begin{indented}
\item[]
\end{indented}

\begin{abstract}
Information geometry promotes an investigation of the geometric structure of statistical manifolds, providing a series of elucidations in various areas of scientific knowledge. In the physical sciences, especially in quantum theory, this geometric method has an incredible parallel with the distinguishability of states, an ability of great value for determining the effectiveness in implementing physical processes. This gives us the context for this work. Here we will approach a problem of uniparametric statistical inference from an information-geometric perspective. We will obtain the generalized Bhattacharyya higher-order corrections for the Cram\'{e}r-Rao bound, where the statistics is given by a mixed quantum state. Using an unbiased estimator $T$, canonically conjugated to the Hamiltonian $H$ that generates the dynamics, we find these corrections independent of the specific choice of estimator. This procedure is performed using information-geometric techniques, establishing connections with corrections to the pure states case.
\end{abstract}

%
%
%
%
%

\section{Introduction}
\label{sec:introducao}

Quantum technologies utilise quantum states as the fundamental agents responsible for information processing. Knowing the quantum operations that can act on these states, their proper control allows an optimization in coding/decoding, manipulation and transmission of the information content on the state of the system. The access to this encoded information, after the operation is applied, demands some way to distinguish the initial and final states of the system \cite{ikemike}. Therefore, distinguishing states in quantum models is a typical task at the core of information theory. Several measures of distance in state spaces are known in the literature that are used as quantifiers of distinguishability, for example, the trace distance, Bures distance, Hellinger distance, Hilbert-Schmidt distance, relative entropies, among others \cite{fuchs96, watrous2018}. As the concept of distance is closely linked to some type of geometric structure of the space involved, there is a natural connection between methods of geometry and information theory.

In this  geometric formulation of information theory, the quantum state space -- a set of density operators that act on the Hilbert space of a quantum system -- is treated as a differentiable manifold equipped with metric tensors that will define the notion of distance between the elements of the manifold. Such an approach allows the development of interesting physical results exploring the metric properties of space \cite{braunstein-caves, wootters81, amari1993infogeo, anderssonheydari1, anderssonheydari2, heydari}. In this way, it is natural to interpret the distinguishability measures as metrics defined in quantum statistical manifolds. This culminates in {\em information geometry}, an area that applies differential geometry methods to the solution and formalization of information science problems, obtaining robust and elegant results with broad applicability \cite{petz1996, gibiliscoisola, petzhasegawa, petz2002, hiaipetz, anderssonheydari3, jarzyna, benyosborne}. Such a framework was used successfully in problems of statistical inference, and it is specifically the one-parameter instance of these problems that will be the focus of our discussion \cite{brody1996prl, brody1996royalsoc}.

The uniparametric quantum statistical inference problem can be described as follows \cite{Helstrom, paris-ijqi, GLM}: Consider a physical system characterized by the pure quantum state $\rho(t)$, that depends on an unknown parameter $t$. We want to estimate the value of this parameter using an unbiased estimator $T$, i.e., $\mbox{E}_ {\rho} [T] = t$. Thus, we can ask: How accurate is this estimation procedure? To answer this question, note that the variance of the estimator can be interpreted as the error associated with the estimation \cite{Helstrom, paris-ijqi, GLM}. Therefore, a good estimation should have a small error, which means that the variance of the estimator should be as small as possible. In literature it is established a lower bound for the variance of the estimator, named the Cram\'{e}r-Rao inequality \cite{Helstrom, paris-ijqi, GLM},
\begin{equation}
\label{desigualdade_cr}
\mbox{Var}_{\rho}[T]\ge \frac{1}{\mathcal{G}},
\end{equation}
where $\mbox{Var}_{\rho}[T]\equiv\Delta T^2$ is the variance of the estimator $T$ and $\mathcal{G}$ is the quantum Fisher information associated to the parameter of interest. From this relation we can see that the more information the state provides about the parameter, the smaller is $\mbox{Var}_ {\rho}[T]$ and the better the estimation \cite{Helstrom, paris-ijqi, GLM, fabricio}.

In this work, we extend the analysis of the Cramér-Rao bound by obtaining its higher-order corrections using the geometry of the quantum state space for {\it mixed} states instead of {\it pure} states. We obtain the following generalized form of the bound in the mixed quantum state scenario:
\begin{prop}
Let $\rho(t)$ be a one-parameter mixed state under a von Neumann dynamics generated by the Hamiltonian $H$, canonically conjugated to the unbiased estimator $T$ of the parameter $t$. The generalized Cram\'{e}r-Rao bound is then given by
\begin{equation}
\label{qntcrlb_classys_correcao_misto}
(\Delta T^2 + \delta T^2)(\Delta H^2 - \delta H^2) \geq \frac{1}{4}\left[1+ \frac{(\mu_4-3\mu_2^2)^2}{\mu_6\,\mu_2-\mu_4^2} \right],
\end{equation}
where $\delta X = \tr(X\sqrt{\rho}X\sqrt{\rho})-[\tr(X\rho)]^2$, and $\mu_{2n}$ are the norms of the $n-$th derivatives of the square root of the state with respect to the parameter.
\end{prop}



The paper is organized as follows. In Section II we introduce the notation used along the manuscript and briefly discuss the current literature on quantum statistical estimation. Section III is devoted to obtain the extension of quantum estimation theory when considering a mixed state scenario, analyzing the consequences for the Cramér-Rao bound. In Section IV we present the higher order corrections for the variance implied by this mixed state scenario. In Section V, we present an algorithmic approach to obtain all possible corrections to the bound, and in Section VI we present our conclusions and discussions.

\section{Notation}
\label{sec:adapt_notacao}

Throughout the paper we use the following notation. Let $\mathcal{B}_{HS}$ be a Hilbert space contained in the set of square matrices with complex entries, equipped with the Hilbert-Schmidt inner product. For a given matrix $A\in \mathcal{B}_{HS}$ we have
\begin{equation*}
\langle A,A\rangle_{HS}=\tr (A^\dagger A)<\infty,
\end{equation*}
where $\langle\bullet,\bullet\rangle_{HS}$ denotes the Hilbert-Schmidt product and $\tr(\bullet)$ stands for the matrix trace. Therefore, $\mathcal{B}_{HS}$ is the set of all matrices with finite Hilbert-Schmidt norm. Thus, let us index the Hilbert-Schmidt product as $g_{ab}$. This definition allows us to establish that $\zeta^ag_{ab}\zeta^b \equiv \tr (\zeta^\dagger\zeta)$, where $\zeta^a \in \mathcal{B}_{HS}$ and $\zeta$ is the matrix associated to $\zeta^a$.

On the other hand, the set of density operators representing mixed states in quantum theory, $\mathcal{S}=\{ \rho \mid \rho=\rho^\dagger > 0, \tr(\rho)= 1, \tr(\rho^{2}) < 1 \}$, denoting a space of positive definite density operators, is a differentiable manifold in $\mathcal{B}_{HS}$ \cite{amari1993infogeo, livro_geo_info}. As the operators in $\mathcal{S}$ are positive definite, we have the embedding $\rho \rightarrow \sqrt{\rho}$ which maps an operator into its own square-root. Thus, identifying $\sqrt{\rho}\equiv\xi$, we find that $g_{ab}\xi^a\xi^b=\tr (\xi\xi)=1$.

A random variable in $\mathcal{S}$ is a form $X_{ab}$ whose average in terms of the state $\xi$ is given by
\begin{equation*}
\mbox{E}_{\xi}[X]=\xi^aX_{ab}\xi^b=\tr [\xi X\xi].
\end{equation*}
Also, the variance of $X_{ab}$ is given by
\begin{equation*}
\mbox{Var}_{\xi}[X]=\Delta X^2=\xi^a\tilde{X}_{ac}\tensor[ ]{\tilde{X}}{^c_b}\xi^b,
\end{equation*}
where $\tilde{X}_{ab} \equiv X_{ab}-g_{ab}(\xi^cX_{cd}\xi^d)$. It is important to note that when calculating these quantities, in general, we have $X_{ac}\xi^c\xi^b\tensor{Y}{_b^a}\neq\xi^cX_{ca}\xi^b\tensor{Y}{_b^a}$ since $\tr [X\xi\xi Y]\neq\tr [\xi X \xi Y]$. Once $\xi^a$ is a matrix, the way the contraction is taken is quite important, even for symmetric forms, because the matrix algebra is non-commutative.

Now consider that $\mathcal{S}$ is parametrically given by a set of local coordinates $\{ \theta = [\theta^{i}] \in \Re^{r}; i=1,\dots,r\}$, where $\xi(\theta)\in C^\infty$. Defining $\partial_i\equiv\partial/\partial\theta^i$ we have, in terms of local coordinates in $\mathcal{S}$, the Riemannian metric
\begin{equation}
\label{fisher_metrica_misto}
G_{ij}=2g_{ab}\dell_i\xi^a\dell_j\xi^b=2\tr [\dell_i\xi\dell_j\xi]
\end{equation}
induced by the Hilbert-Schmidt metric $g_{ab}$ of $\mathcal{B}_{HS}(\Hi)$. The proof of Eq.(\ref{fisher_metrica_misto}) follows the same steps as the proof of proposition 1 in Ref.~\cite{brody1996royalsoc}. The difference is that there, a factor 4 is arbitrarily considered so one has an identification with the results of the Cramér-Rao bound while here we have a factor 2. The motivation for this difference will become clear in the next section.

\section{Quantum statistical estimation}
\label{sec:est_qnt_misto}

In this section, inspired by Ref.\cite{brody2011fasttrack}, we are going to show a generalisation considering the mixed state scenario. Thus, consider that, in $\mathcal{S}$, we have the result of a measurement from an unbiased estimator $T_{ab}$ such that $\xi^aT_{ab}\xi^b=t$, and that a one-parameter family $\xi(t)$ characterizes the probability distribution of all the possible results of the measurement. We can refer to $t$ as the time that has passed since the preparation of the initially known normalized state $\xi_0 = \xi(0)$. Our goal is to estimate $t$.

The system is prepared in the normalized state $\xi_0$  and evolves under the Hamiltonian $H_{ab}$ following the equation of motion
\begin{equation*}
\xi_t = e^{-iHt}\,\xi_0\,e^{iHt}.
\end{equation*}
The derivative of this equation in time gives
\begin{equation*}
\dot{\xi}=-iH\xi+i\xi H = -i[H,\xi]
\end{equation*}
that is the von Neumann dynamics for the state $\xi_t$. From such evolution we find that
\begin{equation*}
\tr (\dot{\xi}\dot{\xi})=2[\tr (H^2\xi\xi)-\tr (H\xi H\xi)].
\end{equation*}

Given the mapping $\xi\rightarrow\sqrt{\rho}$, we find that the right side of the previous expression is twice the Wigner-Yanase skew information (WYSI) \cite{WY} of an arbitrary observable $H$ and a quantum state $\rho$
\begin{equation}
\label{wysi}
I_{\rho}(H)=\tr (H^2\rho)-\tr (H\sqrt{\rho}H\sqrt{\rho})=-\frac{1}{2}\tr ([\sqrt{\rho},H]^2).
\end{equation}
Defining $\delta X^2 \equiv \tr (X \xi X \xi)-[\tr (X \xi \xi)]^2$ we can rewrite Eq.(\ref{wysi}) as
\begin{equation*}
I_{\rho}(H)=(\Delta H^2-\delta H^2).
\end{equation*}
From the embedding $\rho\rightarrow \sqrt{\rho}$, we find that the natural metric, induced by the notion of distance in the ambient manifold, is going to be the WYSI in the mixed state space which is the only metric that has a Riemannian $\alpha-$connection \cite{amari1993infogeo}.


Relaxing the unitarity of the trace, we can define a symmetric function of $\xi$ in $\mathcal{B}_{HS}$ as
\begin{equation*}
t(\xi)=\frac{\xi^aT_{ab}\xi^b}{\xi^cg_{cd}\xi^d}=\frac{\tr (\xi T\xi)}{\tr (\xi\xi)}.
\end{equation*}
Again, being the matrix algebra non-commutative, we must be careful when doing the contractions. For example, considering $\triangledown_c=\dell/\dell\xi^c$, we obtain
\begin{eqnarray*}
\triangledown_c\xi^aT_{ab}\xi^b &=& \triangledown_c\xi^aT_{ab}\xi^b + \xi^aT_{ab}\triangledown_c\xi^b \nonumber \\
&=& T_{cb}\xi^b + \xi^aT_{ac}  \\
&=& (T_{ac}+T_{ca})\xi^a=(T\xi+\xi T).
\end{eqnarray*}
Following this calculation, we can find that the gradient of a function $t$ in $\mathcal{B}_{HS}$ is given by
\begin{equation*}
\triangledown_ct=\frac{(T_{ac}+T_{ca})\xi^a-2\,t\,g_{ac}\xi^a}{\xi^b\xi_b}.
\end{equation*}
Now, imposing the unitarity of the trace, $\tr(\xi\xi)=1$, we get
\begin{eqnarray*}
\triangledown_ct &=& (\tilde{T}_{ac}+\tilde{T}_{ca})\xi^a \\
\triangledown^ct &=& (\tensor{T}{^c_a}+\tensor{T}{_a^c})\xi^a -2t\xi^c.
\end{eqnarray*}
The gradient norm is
\begin{eqnarray*}
\label{norm_gradt_misto}
\norm{\triangledown^ct}^2&=&(T_{ac}+T_{ca})(\tensor{T}{^c_b}+\tensor{T}{_b^c})\xi^a\xi^b  \\
&-&2\,t\,(T_{ac}+T_{ca})\xi^a\xi^c -2\,t\,(T_{ab}+T_{ba})\xi^a\xi^b + 4\,t^2\,\xi^a\xi_a  \\
&=&T_{ca}\xi^a\tensor{T}{^c_b}\xi^b+T_{ca}\xi^a\xi^b\tensor{T}{_b^c}  \\
&+&\xi^aT_{ac}\tensor{T}{^c_b}\xi^b + \xi^aT_{ac}\xi^b\tensor{T}{_b^c}-4t^2  \\
&=&2\{ [\tr (T^2\xi\xi)-t^2] + [\tr (T\xi T\xi)-t^2] \}  \\
&=&2(\Delta T^2 + \delta T^2)
\end{eqnarray*}
The von Neumann dynamics can be written as
\begin{equation*}
\dot{\xi}^a=-i\tensor{H}{^a_b}\xi^b+i\xi^b\tensor{H}{_b^a}.
\end{equation*}
Being $T$ canonically conjugated to $H$, i.e., $i[H,T]=1$, we obtain the projection of $\triangledown^ct$ into the direction $\dot{\xi}^a$
\begin{eqnarray*}
\triangledown_at\,\dot{\xi}^a&=&[(T_{ac}+T_{ca})\xi^c-2\,t\,\xi_a][-i\tensor{H}{^a_b}\xi^b+i\xi^b\tensor{H}{_b^a}] \nonumber \\
&=&\tr (-iT\xi H\xi+i\xi T\xi H+iT\xi\xi H-i\xi TH\xi) \nonumber \\
&=&i\tr ([H,T]\xi\xi)=1.
\end{eqnarray*}

Using the Cauchy-Schwartz inequality for a pair of hermitian operators $X$ and $Y$
\begin{equation*}
[\tr (XY)]^2\geq\tr(X^2)\tr(Y^2),
\end{equation*}
we deduce the quantum Cramér-Rao inequality for the density operators space
\begin{equation}
\label{qntcrlb_misto2}
(\Delta T^2+\delta T^2)(\Delta H^2-\delta H^2)\geq\frac{1}{4},
\end{equation}
where $\Delta X$ are the variances of the parameters $X$ and $(\Delta X^2 + \delta X^2)$ is the skew information of second kind \cite{brody2011fasttrack}. Note that the relation above extends the usual notion of Cram\'{e}r-Rao bound, considering some corrections to the uncertainty relation.


If we impose that $\xi=\xi^2$ to recover the pure state case, it follows that $\tr (H\sqrt{\rho}H\sqrt{\rho}) = [\tr (H\rho)]^2$, consequently $\delta H^2 = 0$, and thus $\tr (\dot{\xi}\dot{\xi})=4\Delta H^2$. Similarly, we can find that $\delta T^2 = 0$. After these considerations, the inequality reduces to the Cramér-Rao bound for pure states found in literature \cite{brody1996prl, brody1996royalsoc}\footnote{In the quantum scenario, the metric when $\alpha=0$ corresponds exactly to the WYSI in the uniparametric case, where the factor 4 can be found \cite{amari1993infogeo}. However, the embedding considered in this case is $\rho\rightarrow 2\sqrt{\rho}$, which is different in the present work. The factor 2 allows us to establish an equality with the relation in Eq.(\ref{qntcrlb_misto2}), in the sense of the Cramér-Rao $1/\mathcal{G}$, and recover the expression for the pure state case. Notice that if we consider the embedding $\rho\rightarrow 2\sqrt{\rho}$, the metric will have the factor 4 and, when taking Eq.(\ref{qntcrlb_misto2}) in the pure state case, we would find $\Delta T^2\Delta H^2\ge 1/8$ \cite{ref20_brody_fasttrack}.}.

We can go a step further and rewrite the inequality given in Eq.(\ref{qntcrlb_misto2}), after a simple manipulation, in another form
\begin{equation*}
\Delta T^2\Delta H^2 \geq \frac{1}{4} + \delta T^2\delta H^2.
\end{equation*}
However, although more symmetric, this inequality is less tight then the previous one.

\section{Higher orders corrections for the variance bound in mixed state case}
\label{sec:ordem_sup_var_qnt}

The dynamics of a mixed states will not lead to an exponential family of states that saturate the Cauchy-Schwartz inequality, i.e., Eq.(\ref{qntcrlb_misto2}) has its minimum bound unattainable when states evolve under a von Neumann dynamics. The calculation of higher-order corrections then becomes relevant to find how the estimation can be affected in such circumstances.

In order to establish the higher-order corrections to the bound in the mixed state scenario, we will be inspired by the works of Bhattacharyya \cite{bhatt1, bhatt2, bhatt3}, we will follow the programme outlined in Ref.\cite{brody2011fasttrack} and finding in this new context what is called generalized Bhattacharyya bound \cite{brody1996prl, brody1996royalsoc}.

\begin{prop}\label{prop:bhattacharrya_misto} Let
\begin{equation*}
\hat{\xi}^{(n)a}=\xi^{(n)a}-\frac{\xi^{(n-1)b}\xi^{(n)}_b}{\xi^{(n-1)c}\xi^{(n-1)}_c}\,\xi^{(n-1)a} - \dots - (\xi^b\xi^{(n)}_b)\,\xi^a \quad (n=0,1,2,\dots),
\end{equation*}
be an orthogonal vectors system, where $\xi^{(n)a}=d^n\xi^a/dt^n$ and
\begin{eqnarray*}
\hat{\xi}^{(r)}_a\xi^a&=&0 \nonumber \\
\hat{\xi}^{(r)}_a\xi^{(s)a} &=& 0; r\neq s \nonumber \\
\hat{\xi}^{(r)}_a\hat{\xi}^{(s)a} &=& 0; r\neq s,
\end{eqnarray*}
are defined for mixed states. The generalized Bhattacharyya lower bounds for an unbiased estimator $T_{ab}$ of a function $t$ can be expressed in the form
\begin{equation}
\label{bhattacharrya_misto}
\Delta T^2 + \delta T^2 \geq \frac{1}{2}\sum_n\frac{(\triangledown_at\,\hat{\xi}^{(n)a})^2}{\hat{\xi}^{(n)a}\hat{\xi}^{(n)}_a}
\end{equation}
\end{prop}
\begin{proof}[Proof]
Let $\hat{T}_{ab}\equiv (T_{ab}+T_{ba})-2\,t\,g_{ab}=(\tilde{T}_{ab}+\tilde{T}_{ba})$. Defining the tensor $R_{ab}\equiv \hat{T}_{ab}+\sum_n\lambda_n\xi_{(a)}\hat{\xi}_{(b)}$, we obtain the variance for $R$
\begin{equation*}
\mbox{Var}_{\xi}[R] = \mbox{Var}_{\xi}[\hat{T}]+\sum_n\lambda_n(\xi^a\hat{T}_{ac}\hat{\xi}^{(n)c}+\hat{\xi}^{(n)c}\hat{T}_{ca}\xi^a) + \sum_n\lambda_r^2\hat{\xi}^{(n)b}\hat{\xi}^{(n)}_b.
\end{equation*}
To obtain each value $\lambda_n$ that minimizes $\mbox{Var}_{\xi}[R]$, we need to consider the variance as a function of $\{\lambda_n\}_{n\in\N}$ variables and find the extreme points that cause the gradient of this function to vanish. There is only one extreme point, a minimum that presents all entries with the same value, namely $\lambda_n^{min}=-(\xi^a\hat{T}_{ac}\hat{\xi}^{(n)c}+\hat{\xi}^{(n)c}\hat{T}_{ca}\xi^a)/2\hat{\xi}^{(n)}_b\hat{\xi}^{(n)b}$ for all $n$. Replacing this in the above expression, we find
\begin{equation*}
\min_{\{\lambda_n\}_{n\in\N}}\left(\mbox{Var}_{\xi}[R]\right) = \mbox{Var}_{\xi}[\hat{T}]-\sum_n\frac{(\xi^a\hat{T}_{ac}\hat{\xi}^{(n)c}+\hat{\xi}^{(n)c}\hat{T}_{ca}\xi^a)^2}{4\hat{\xi}^{(n)b}\hat{\xi}^{(n)}_b}.
\end{equation*}
Since $\mbox{Var}_{\xi}[R]\geq 0$, we obtain an expression for the generalized bound on the variance of $\hat{T}_{ab}$,
\begin{equation*}
\mbox{Var}_{\xi}[\hat{T}]\geq \sum_n\frac{(\xi^a\hat{T}_{ac}\hat{\xi}^{(n)c}+\hat{\xi}^{(n)c}\hat{T}_{ca}\xi^a)^2}{4\hat{\xi}^{(n)b}\hat{\xi}^{(n)}_b}.
\end{equation*}
We need to obtain
\begin{equation*}
(\xi^a\hat{T}_{ac}\hat{\xi}^{(n)c}+\hat{\xi}^{(n)c}\hat{T}_{ca}\xi^a)^2=4(\xi^aT_{ac}\hat{\xi}^{(n)c}+\hat{\xi}^{(n)c}T_{ca}\xi^a)^2=4(\triangledown_at\,\hat{\xi}^{(n)a})^2
\end{equation*}
and
\begin{equation*}
\mbox{Var}_{\xi}[\hat{T}] = 2(\Delta T^2+\delta T^2).
\end{equation*}
Combining all these results we arrive at the desired expression in Eq.(\ref{bhattacharrya_misto}). Naturally, for the case $r=1$, we recover the inequality given in Eq.(\ref{qntcrlb_misto2}).
\end{proof}

A simple interpretation of this proposition is that, given the gradient vector $\triangledown^at$, its squared modulus will always be greater than or equal to the sum of the squares of its orthogonal components with respect to a given base. Applying Cauchy-Schwarz inequality is the same as using order$-1$ Bhattacharyya inequality. Note that the generalized Bhattacharyya bound is not necessarily independent of the specific choice of the estimator $T_{ab}$. This is evidence that they are not fully equivalent to the original Bhattacharrya bounds, not even at the classical level, since the original bounds are independent of the specific choice of estimator. Therefore, we want to obtain corrections which are independent of that choice, so that the bound will not depend on the way we perform the estimation. This will again demand $T$ and $H$ to be canonically conjugated.

Before focusing on higher order corrections, let us present some useful results. These results are adaptations of statements from \cite{brody1996royalsoc} to the context of mixed states.

\begin{prop}
\label{lem:norm_xin_indep_t_misto}
Given $\xi^a(t)$ satisfying the von Neumann dynamics, the norm of $\xi^{(n)a}$, $g_{ab}\xi^{(n)a}\xi^{(n)b}$, is independent of the parameter $t$, where $\xi^{(n)a}=d^n\xi^a/dt^n$. In particular, $g_{ab}\xi^{(n)a}\xi^{(n+1)b}=0$.
\end{prop}
\begin{proof}[Proof]
If von Neumann is valid, then $\dot{\xi}^{(n)a}=-i\tensor{H}{^a_b}\xi^{(n)b}+i\xi^{(n)b}\tensor{H}{_b^a}$. The time derivative of the norm of $\xi^{(n)a}$ gives
\begin{eqnarray*}
\frac{d}{dt}\left[g_{ab}\xi^{(n)a}\xi^{(n)b}\right]&=&2g_{ab}\xi^{(n)a}\dot{\xi}^{(n)b}=2g_{ab}\xi^{(n)a}\xi^{(n+1)b} \nonumber \\
&=&\xi^{(n)}_b(-i\tensor{H}{^a_c}\xi^{(n)c}+i\xi^{(n)c}\tensor{H}{_c^b}) \\
&=&\tr (-i\xi^{(n)}H\xi^{(n)}+i\xi^{(n)}H\xi^{(n)})=0,
\end{eqnarray*}
completing the proof.
\end{proof}

It is important to note that since the von Neumann dynamics is given by $\dot{\xi}=-i[\tilde{H},\xi]$, where $\tilde{H}_{ab}=H_{ab}-g_{ab}[\tr (H\xi\xi)]$, we can generalize this relation to
\begin{equation}
\label{vonneumann_nderiv}
\xi^{(n)}=\mbox{mod}_{-i}[n]\cdot\mbox{Ad}^n_{\tilde{H}}[\xi],
\end{equation}
where
\begin{eqnarray*}
\mbox{mod}_{-i}[n] \Longrightarrow && n=\bar{1}\rightarrow-i \\
&& n=\bar{2}\rightarrow -1 \\
&& n=\bar{3}\rightarrow i \\
&& n=\bar{4}\rightarrow +1
\end{eqnarray*}
and
$$\mbox{Ad}^n_{\tilde{H}}[\bullet]=[\dots[\tilde{H},[\tilde{H},[\tilde{H},[\tilde{H},\bullet]]]]\dots].$$
Let us also define the $n-$th derivative of the norm of $\xi$ as
\begin{equation*}
g_{ab}\xi^{(n)a}\xi^{(n)b}\equiv \mu_{2n}.
\end{equation*}

\begin{prop}
\label{prop:conseq_TconjugH_misto}
Let $T_{ab}$ be canonically conjugate to $H_{ab}$ and an unbiased estimator to the parameter $t$. Then
\begin{equation}
\label{conseq_TconjugH_misto}
T_{ab}\xi^{(n)a}\xi^{(n)b}=t\,g_{ab}\xi^{(n)a}\xi^{(n)b}+\kappa,
\end{equation}
where $\kappa$ is a constant. Therefore, for all $n$, $\tilde{T}_{ab}\xi^{(n)a}\xi^{(n)b}=\kappa$ is a constant of motion along the path $\xi(t)$, following a von Neumann dynamics.
\end{prop}
\begin{proof}[Proof]
If $T$ is canonically conjugated to $H$, thus  $ig_{ab}=(T_{ac}\tensor{H}{^c_b}-\tensor{H}{_a^c}T_{cb})$. It follows that
\begin{equation*}
g_{ab}\xi^{(n)a}\xi^{(n)b}=(-i\xi^{(n)a}T_{ac}\tensor{H}{^c_b}\xi^{(n)b}+i\xi^{(n)a}\tensor{H}{_a^c}T_{cb}\xi^{(n)b})
\end{equation*}
The derivative of the average of $T$ in the state $\xi^{(n)}$ gives
\begin{eqnarray*}
\frac{d}{dt}[T_{ab}\xi^{(n)a}\xi^{(n)b}]&=&T_{ab}(\dot{\xi}^{(n)a}\xi^{(n)b}+\xi^{(n)a}\dot{\xi}^{(n)b}) \\
&=&(-i\xi^{(n)a}T_{ac}\tensor{H}{^c_b}\xi^{(n)b}+i\xi^{(n)a}\tensor{H}{_a^c}T_{cb}\xi^{(n)b}) \\
&=&g_{ab}\xi^{(n)a}\xi^{(n)b}.
\end{eqnarray*}
Integrating the above expression, we obtain
\begin{eqnarray*}
T_{ab}\xi^{(n)a}\xi^{(n)b}&=&t\,g_{ab}\xi^{(n)a}\xi^{(n)b}+\kappa  \\
\tilde{T}_{ab}\xi^{(n)a}\xi^{(n)b}&=&\kappa,
\end{eqnarray*}
where $\kappa$ is a constant independent of $t$ and $\tilde{T}_{ab}\xi^{(n)a}\xi^{(n)b}$ is a constant of motion.
\end{proof}

\begin{prop}
\label{lem:2Txi(n)xi_misto}
Let $T_{ab}$ be canonically conjugate to $H_{ab}$ and an unbiased estimator to the parameter $t$. Thus, considering $n$ odd integers, with  $m=(n-1)/2$, we have
\begin{equation}
\label{2Txi(n)xi_misto}
(T_{ab}+T_{ba})\xi^{(n)a}\xi^b=(-1)^mng_{ab}\xi^{(m)a}\xi^{(m)b}
\end{equation}
\end{prop}
\begin{proof}[Proof]
Given Proposition \ref{prop:conseq_TconjugH_misto}, the proof follows the one presented in Lemma 6 of Ref.~\cite{brody1996royalsoc}.
\end{proof}

In view of these results, we can deduce some higher order corrections, independent of the specific choice of $T$, for canonically conjugated observables in the mixed state case.

\begin{proof}[Proof of Proposition 1]
Let us consider corrections up to the third order for the bound in estimation of mixed states that emerge when we expand $\triangledown_at$ over the orthogonal vector system $\dot{\xi}^a$, $\hat{\xi}^{(2)a}$ and $\hat{\xi}^{(3)a}$. From Eq.(\ref{bhattacharrya_misto}), the generalized bound is
\begin{equation}
\label{cota_classysvar_ordem3_misto}
\Delta T^2+\delta T^2\geq \frac{(\dot{\xi}^a\triangledown_at)^2}{2\dot{\xi}^b\dot{\xi}_b}+\frac{(\hat{\xi}^{(2)a}\triangledown_at)^2}{2\hat{\xi}^{(2)b}\hat{\xi}^{(2)}_b}+\frac{(\hat{\xi}^{(3)a}\triangledown_at)^2}{4\hat{\xi}^{(2)b}\hat{\xi}^{(2)}_b},
\end{equation}
where $\hat{\xi}^{(2)a}$ and $\hat{\xi}^{(3)a}$ are given by
\begin{equation*}
\hat{\xi}^{(2)a}=\ddot{\xi}^a-\frac{(\ddot{\xi}^b\dot{\xi}_b)}{(\dot{\xi}^c\dot{\xi}_c)}\dot{\xi}^a-(\ddot{\xi}^b\xi_b)\xi^a
\end{equation*}
and
\begin{equation*}
\hat{\xi}^{(3)a}=\dddot{\xi}^a-\frac{(\dddot{\xi}^b\ddot{\xi}_b)}{(\ddot{\xi}^c\ddot{\xi}_c)}\ddot{\xi}^a-\frac{(\dddot{\xi}^b\dot{\xi}_b)}{(\dot{\xi}^c\dot{\xi}_c)}\dot{\xi}^a-(\dddot{\xi}^b\xi_b)\xi^a
\end{equation*}

The term corresponding to $r=1$ in Eq.(\ref{cota_classysvar_ordem3_misto}) has already been calculated in the derivation of Eq.(\ref{qntcrlb_misto2}). Thus let us proceed with the second order term. The expression in Eq.(\ref{vonneumann_nderiv}) and Prop.~\ref{lem:norm_xin_indep_t_misto} give us that $\ddot{\xi}=-[\tilde{H},[\tilde{H},\xi]]$, $\ddot{\xi}^a\dot{\xi}_a=0$, $\ddot{\xi}^a\xi_a=-2I_{\rho}(\tilde{H})$. Therefore, $\hat{\xi}^{(2)a}=2I_{\rho}(\tilde{H})\xi-[\tilde{H},[\tilde{H},\xi]]$. Hence, the numerator of the second order term is $(\hat{\xi}^{(2)a}\triangledown_at)^2=\big\{\tr \big[(T\xi-\xi T)(2I_{\rho}(\tilde{H})\xi-[\tilde{H},[\tilde{H},\xi]])\big]\big\}^2$ which explicitly depends on the choice of the estimator $T$. Due to this dependence, we will discard this term, as we want only correction terms that do not depend on the choice of the estimator.

Now, dealing with the term that involves $\hat{\xi}^{(3)a}$, we have $\dddot{\xi}^a\ddot{\xi}_a=\ddot{\xi}^a\dot{\xi}_a=0$ from Prop.~\ref{lem:norm_xin_indep_t_misto}. It follows that $\ddot{\xi}^a\xi_a=-\dot{\xi}^a\dot{\xi}_a$, thus $\dddot{\xi}^a\xi_a=0$. From the expression in Eq.(\ref{vonneumann_nderiv}), $\dddot{\xi}^a=i[\tilde{H},[\tilde{H},[\tilde{H},\xi]]]=i([\tilde{H}^3,\xi]+3[\tilde{H}\xi\tilde{H},\tilde{H}])$. After some manipulation, we find $\dddot{\xi}^a\dot{\xi}_a=-2\tr (\tilde{H}^4\xi\xi-4\tilde{H}^3\xi\tilde{H}\xi+3\tilde{H}^2\xi\tilde{H}^2\xi)$. Consequently,
\begin{equation*}
\hat{\xi}^{(3)}=i\left\{ \big([\tilde{H}^3,\xi]+3[\tilde{H}\xi\tilde{H},\tilde{H}]\big) - \frac{\mu_4}{\mu_2}[\tilde{H},\xi]  \right\}
\end{equation*}
and
\begin{equation}
\label{norm_xihat3_misto}
\hat{\xi}^{(3)a}\hat{\xi}^{(3)}_a=\mu_6-\frac{\mu_4^2}{\mu_2},
\end{equation}
where we explicitly have
\begin{eqnarray*}
\mu_2&: =&2I_{\rho}(\tilde{H})=2\tr (\tilde{H}^2\xi\xi-\tilde{H}\xi\tilde{H}\xi)=2(\Delta H^2 + \delta H^2)=\tr(\dot{\xi}\dot{\xi}) \\
\mu_4&: =& 2\tr (\tilde{H}^4\xi\xi-4\tilde{H}^3\xi\tilde{H}\xi+3\tilde{H}^2\xi\tilde{H}^2\xi) =\tr(\ddot{\xi}\ddot{\xi})\\
\mu_6&: =& 2\tr (\tilde{H}^6\xi\xi-6\tilde{H}^5\xi\tilde{H}\xi+15\tilde{H}^4\xi\tilde{H}^2\xi-10\tilde{H}^3\xi\tilde{H}^3\xi)=\tr(\xi^{(3)}\xi^{(3)}).
\end{eqnarray*}

Regarding the numerator of the third order term
\begin{equation*}
\hat{\xi}^{(3)a}\triangledown_at=i\tr \left\{ [T,\tilde{H}^3]\xi\xi+3(\xi[\tilde{H}^2,T]\xi\tilde{H}+\xi[T,\tilde{H}]\xi\tilde{H}^2) -\xi[T,\tilde{H}]\xi\frac{\mu_4}{\mu_2} \right\},
\end{equation*}
we see that it involves commutators between $T$ and $H^k,k\in\N$. By finite induction one shows that $[\tilde{H}^k,T]=-ki\tilde{H}^{k-1}, k\in\N$, and we have
\begin{equation}
\label{gradt_xihat3_misto}
\hat{\xi}^{(3)a}\triangledown_at=\frac{\mu_4^2}{\mu_2}-3\mu_2.
\end{equation}
Due to Prop.~\ref{lem:2Txi(n)xi_misto}, which hypothesizes the canonical conjugation between $T$ and $H$, in general, even terms explicitly depend on the arbitrary choice of the estimator $T$ and odd terms involve commutators between $T$ and $H ^ k, k \in \N$; therefore the odd terms will not depend on the choice of $T$. Thus, putting the results in Eqs.(\ref{qntcrlb_misto2}), (\ref{norm_xihat3_misto}), and (\ref{gradt_xihat3_misto}) together, we obtain the following Heisenberg-like correction for the Cramér-Rao bound in the mixed state case, which only depends on the dynamics of $\xi(t)$ generated by $H$, i.e., the statement of Proposition 1.


Imposing $\xi=\xi^2$ to recover the pure state case, we have $\mu_4 \neq \langle\tilde{H}^4\rangle$, where $\langle\tilde{H}^n\rangle=\tr(\tilde{H}^n\xi\xi)$ denotes the $n$-th moment of the Hamiltonian in the corresponding state $\xi$. The only circumstance where the equality holds is when $\mu_2=\langle\tilde{H}^2\rangle$. Thus, just imposing that the density operator characterizes a pure state Eq.(\ref{qntcrlb_classys_correcao_misto}) does not the recover the results in Refs.~\cite{brody1996prl, brody1996royalsoc}, contrary to what was expected \cite{brody2011fasttrack}.
\end{proof}

\section{Method to obtain higher-orders corrections}
\label{sec:algorit_ordem_impar}

In the previous section, we found that \textit{even} higher order corrections explicitly depend on the choice of the estimator $T$, while by Prop.~\ref{lem:2Txi(n)xi_misto} \textit{odd} higher order corrections are independent of the specific choice of $T$. We will then present an algorithmic way of calculating odd terms of higher order corrections following the steps seen in Ref.~\cite{algorit_ordem_impar} for a one-parameter family of states $\xi(t)$ evolving under von Neumann dynamics.

Let the series of orthogonal vectors be
\begin{equation*}
\left\{ \xi^a, \dot{\xi}^a, \ddot{\xi}^a - (\ddot{\xi}^b\xi_b)\xi^a, \dddot{\xi}^a-\frac{\dddot{\xi}^b\dot{\xi}_b}{\dot{\xi}^c\dot{\xi}_c}\dot{\xi}^a, \dots \right\},
\end{equation*}
that are basically the $\hat{\xi}^{(n)a}$, given in Prop.~\ref{prop:bhattacharrya_misto}, dismissing the vanishing terms (e.g. $\dddot{\xi}^a\ddot{\xi}_a = \ddot{\xi}^a\dot{\xi}_a = \ddot{\xi}^a\xi_a=0$). Let us denote this series of vectors by $\{\uppsi^a_n\}$; thus $\uppsi^a_0=\xi^a$, $\uppsi^a_1=\dot{\xi}^a$ and so on.

For each \textit{odd integer value\/} $n$ the basis vectors $\uppsi^a_n$ are obtained by subtracting the components $\uppsi^a_k$ of $\xi^{(n)a}$ with $k<n$
\begin{eqnarray*}
\uppsi^a_1 &=& \dot{\xi}^a \\
\uppsi^a_3 &=&\xi^{(3)a}-\frac{\xi^{(3)b}\uppsi_{1b}}{\uppsi^c_1\uppsi_{1c}}\uppsi^a_1 \\
\uppsi^a_5 &=& \xi^{(5)a}-\frac{\xi^{(5)b}\uppsi_{3b}}{\uppsi^c_3\uppsi_{3c}}\uppsi^a_3 - \frac{\xi^{(5)b}\uppsi_{1b}}{\uppsi^c_1\uppsi_{1c}}\uppsi^a_1.
\end{eqnarray*}
Note that $\triangledown_at=(\tilde{T}_{ab}+\tilde{T}_{ba})\xi^b$, so we can rewrite the generalized bounds in Eq.(\ref{bhattacharrya_misto}) as
\begin{equation}
\label{bhattacharrya_misto_uppsi}
\Delta T^2 + \delta T^2 \geq \frac{1}{2}\sum_n\frac{[\uppsi^a_n(\tilde{T}_{ab}+\tilde{T}_{ba})\xi^b]^2}{g_{cd}\uppsi^c_n\uppsi^d_n}.
\end{equation}
Let $N_n \equiv g_{ab}\uppsi^a_n\uppsi^b_n$ stand for the denominator of the correction terms in Eq.(\ref{bhattacharrya_misto_uppsi}). We have
\begin{equation*}
N_n=\frac{D_{2n}}{D_{2n-4}} , \quad n>2,
\end{equation*}
where $D_{2n}$ is defined by the determinant
\begin{equation*}
D_{2n} =
\left| \matrix{ \mu_{2n}& \mu_{2n-2}& \cdots& \mu_{n+1} \cr
\mu_{2n-2}& \mu_{2n-4}& \cdots& \mu_{n-1} \cr
\vdots         & \vdots           & \ddots & \vdots \cr
\mu_{n+1}  & \mu_{n-1}     & \cdots & \mu_{2} \cr} \right|.
\end{equation*}
As examples we found that $D_{2}=\mu_2$, $D_6=\mu_6\mu_2-\mu_4^2$. Note that $N_1=\mu_2=2I_{\rho}(\tilde{H})$. The statistical identities \cite{stuart_kendall} guarantee that $D_{2n}\ge 0$.

Before moving on to the numerator, we define
\begin{equation*}
F_{n,k}\equiv \frac{\xi^{(n)a}\uppsi_{ka}}{\uppsi^b_k\uppsi_{kb}},
\end{equation*}
which has an expression in terms of determinants given by
\begin{equation*}
F_{n,k}=\frac{(-1)^{\frac{1}{2}(n+k)-1}}{D_{2k}}
\left| \matrix{\mu_{n+k} & \mu_{n+k-2} & \dots & \mu_{n+1}\cr
\mu_{2k-2} & \mu_{2k-4} & \dots & \mu_{k-1}\cr
\vdots & & \ddots & \vdots \cr
\mu_{k+1} & \mu_{k-1} & \dots & \mu_{2} \cr} \right| .
\end{equation*}
For example, we have for $k=1,3$
\begin{eqnarray*}
F_{n,3} &=& (-1)^{m+1}\frac{1}{D_6}
\left| \matrix{\mu_{n+3} & \mu_{n+1} \cr
\mu_4 & \mu_2 \cr} \right| \\
F_{n,1} &=& (-1)^m\frac{1}{D_2}\mu_{n+1},
\end{eqnarray*}
where $m=1/2(n-1)$.

We now have all the identities necessary to find a recursive relationship to obtain the odd higher order corrections. Going back to the numerator, we define $U_n\equiv \uppsi^a_n(\tilde{T}_{ab}+\tilde{T}_{ba})\xi^b$. Using $F_{n,k}$ follows the expression for $\uppsi^a_n$
\begin{equation*}
\uppsi^a_n=\xi^{(n)a}-\sum\limits^{n-2}_{k=1,3,5,\dots}F_{n,k}\uppsi^a_k;
\end{equation*}
from the relation above and Prop.~\ref{lem:2Txi(n)xi_misto} follows a recursive formula for $U_n$
\begin{equation}
\label{recursiva_Un}
U_n=(-1)^m\,n\,\mu_{n-1}-\sum\limits^{n-2}_{k=1,3,5,\dots}F_{n,k}U_k,
\end{equation}
where $U_1=1$. Thus, the uncertainty relation in Eq.(\ref{bhattacharrya_misto_uppsi}) can be rewritten as
\begin{equation}
\label{correcao_superior_impares}
(\Delta T^2+\delta T^2)(\Delta H^2-\delta H^2)\ge \frac{1}{4}\sum\limits^{n-2}_{k=1,3,5,\dots}\frac{\mu_2U^2_k}{N_k}.
\end{equation}

Using the relations given in Eqs.(\ref{recursiva_Un}) and (\ref{correcao_superior_impares}), after some algebraic manipulations, we can obtain any correction of a higher order purely in terms of $\mu_{2k}$, regardless of the choice of estimator.

\section{Conclusion}
\label{sec:conclusao}

In this work we approached the quantum statistical estimation problem using mixed states, illustrated by time-energy uncertainty relations, from a geometric perspective. This geometrical view of the space of states allows us to make a clear distinction from the pure state scenario previously reported in literature. We obtained the Cram\'{e}r-Rao bound independent of the choice  of the estimator and analyzed the corrections that emerge naturally. A methodology to obtain higher-order corrections is also presented and, consequently, the path for extensions of the main result is proposed. Contrarily to what was expected, when we imposed $\rho = \rho^2$, there was no reduction to the known bound found for the pure state case.

It is important to note that the square root embedding used in the present work, $\rho\rightarrow \sqrt{\rho}$, is related to the WYSI metric which, in turn, has its Riemannian connection identified with the single $\alpha$-connection of the same type. This fact makes the WYSI a good choice of metric to work over a geometric structure of the state space, since it presents this privileged Riemannian structure from the point of view of $\alpha$-connections. The choice of this embedding was also motivated by the possibility of recovering the $1/4$ factor in the Cram\'{e}r-Rao bound.

To conclude our discussions in the context of single parameter estimation, it is important to remember that the saturation of the Cram\'{e}r-Rao bound can only be achieved when considering two important features: $(i)$ the asymptotic limit of a large number of probes and $(ii)$ performing an optimal measurement given by the eigenbasis of the symmetric logarithmic derivative. Usually in a laboratory, where the experimentalist has access to a limited number of probes, corrections to the bound gain importance to provide tighter estimates to the attainable estimation precision. Here we investigated these corrections approaching the problem from an information-geometric point of view and we obtained that in the context of mixed state quantum estimation, the Wigner-Yanase skew information constitutes a natural metric in the space of states. Next, higher-order corrections obtained to the Cram\'{e}r-Rao inequality based on such metric determine the tightness of the bound for practical purposes. Our work provides advances towards understanding the mixed state estimation paradigm and its practical realization in quantum sensing and metrology.

\section{Note added}
During the preparation of this manuscript, we became aware of related work by A.~J.~Belfield and D.~C.~Brody \cite{brody2020} where higher-order corrections to quantum estimation bounds based on the Wigner-Yanase skew information metric are also discussed.

\ack We thank Dorje C. Brody for careful reading of the manuscript and for fruitful discussions. The project was funded by Brazilian funding agencies CNPq (Grant No. 307028/2019-4), FAPESP (Grant No. 2017/03727-0), Coordena\c{c}\~{a}o de Aperfei\c{c}oamento de Pessoal de N\'{i}vel Superior - Brasil (CAPES) (Finance Code 001), by the Brazilian National Institute of Science and Technology of Quantum Information (INCT/IQ), and by the European Research Council (ERC StG GQCOP Grant No.~637352).


\end{document}